\documentclass[journal]{article}

\usepackage{graphicx,url}
\usepackage{dcolumn}
\usepackage{bm}
\usepackage{amsmath,amsfonts,amssymb}

\newtheorem{theorem}{\noindent{\it Theorem}}[section]

\newtheorem{lemma}[theorem]{\noindent{\it Lemma}}

\newtheorem{remark}[theorem]{\noindent{\it Remark}}
\newtheorem{corollary}[theorem]{\noindent{\it Corollary}}
\newenvironment{proof}{\noindent{\it Proof:}}{$\hfill$ $\Box$\\ }

\begin{document}

\title{Good and asymptotically good quantum codes derived from algebraic geometry}

\author{Giuliano G. La Guardia, Francisco Revson F. Pereira
\thanks{Giuliano Gadioli La Guardia (corresponding author) is with Department of Mathematics and Statistics,
State University of Ponta Grossa (UEPG), 84030-900, Ponta Grossa,
PR, Brazil, e-mail: (gguardia@uepg.br). Francisco Revson F. Pereira
is with Department of Electrical Engineering, Federal University of
Campina Grande (UFCG), 58429-900, Campina Grande, PB, Brazil,
e-mail: (francisco.pereira@ee.ufcg.edu.br).}}

\maketitle

\begin{abstract}
In this paper, we construct several new families of quantum codes
with good parameters. These new quantum codes are derived from
(classical) $t$-point ($t\geq 1$) algebraic geometry (AG) codes by
applying the Calderbank-Shor-Steane (CSS) construction. More
precisely, we construct two classical AG codes $C_1$ and $C_2$ such
that $C_1\subset C_2$, applying after the well-known CSS
construction to $C_1$ and $C_2$. Many of these new codes have large
minimum distances when compared with their code lengths as well as
they also have small Singleton defects. As an example, we construct
a family ${[[46, 2(t_2 - t_1), d]]}_{25}$ of quantum codes, where
$t_1 , t_2$ are positive integers such that $1 <t_1 < t_2 < 23$ and
$d\geq \min \{ 46 - 2t_2 , 2t_1 - 2 \}$, of length $n=46$, with
minimum distance in the range $2\leq d\leq 20$, having Singleton
defect at most four. Additionally, by applying the CSS construction
to sequences of $t$-point (classical) AG codes constructed in this
paper, we generate sequences of asymptotically good quantum codes.
\end{abstract}

\section{Introduction}

Methods and techniques of constructing quantum codes with good
parameters are extensively investigated in the literature
\cite{steane:1996,calderbank:1998,steane:1999,chen:2001,ling:2001,ketkar:2006,aly:2007,laguardia:2009,laguardia:2011,jin:2012,laguardia:2012,laguardia:2013,laguardia:2014,jin:2014,munuera:2016}.
Many of these works
\cite{calderbank:1998,ketkar:2006,aly:2007,laguardia:2011,jin:2012,laguardia:2014}
were performed by applying one (or all) of the following techniques:
1) the well-known CSS construction based on (classical linear)
Euclidean self-orthogonal codes or even based on two (classical
linear) nested codes
\cite{calderbank:1998,ketkar:2006,aly:2007,jin:2012,laguardia:2014};
2) the Hermitian construction applied to (classical linear)
Hermitian self-orthogonal codes
\cite{calderbank:1998,ketkar:2006,aly:2007,laguardia:2011,jin:2012,laguardia:2014,jin:2014};
3) the Steane's enlargement of CSS construction applied to linear
Euclidean self-orthogonal codes
\cite{steane:1996,steane:1999,laguardia:2009,laguardia:2014}. In
particular, the CSS construction was also utilized in chains of
nested classical linear codes to construct quantum codes whose
parameters are asymptotically good
\cite{ashikhmin:2001,chen:2001,ling:2001,matsumoto:2002,walker:2008}.
All these latter asymptotically good quantum codes were constructed
by applying the CSS construction to families of AG codes. In fact,
the class of AG codes is a good source to obtain asymptotically good
codes (see for example \cite{henning:1996,henning:2005}). In Refs.
\cite{ashikhmin:2001,chen:2001,ling:2001,matsumoto:2002}, the
authors constructed asymptotically good binary quantum codes and, in
Ref. \cite{walker:2008}, the authors presented families of nonbinary
asymptotically good quantum codes by means of one-point AG codes.

In this paper, we construct (classical) $t$-point ($t \geq 1$) AG
codes (which are a generalization of one-point AG codes) as well as
AG codes whose divisor $G$ is not a rational place, after applying
the CSS construction to these codes, in order to generate nonbinary
quantum codes with good parameters. Additionally, we also construct
sequences of (classical) $t$-point AG codes to obtain sequences of
asymptotically good quantum codes by means of the CSS construction.
These new constructions presented here are natural generalizations
of the works dealing with constructions of quantum codes derived
from one-point AG codes (see for example
\cite{ashikhmin:2001,chen:2001,ling:2001,matsumoto:2002}).

The paper is arranged as follows. In Section~\ref{sec2}, we recall
the concepts utilized in this work. Section~\ref{sec3} deals with
the contributions of this paper, i.e., constructions of quantum
codes with good and asymptotically good parameters derived from
classical AG codes. In Section~\ref{sec4}, we compare the new code
parameters with the ones shown in the literature and, in
Section~\ref{sec5}, we give a summary of this work.

\section{Preliminaries}\label{sec2}

\subsection{Algebraic geometry codes}\label{sec2.1}

In this subsection we recall the concept of algebraic geometry codes
as well as results that will be utilized in our constructions. Since
this topic of research is not so common, we present it with more
details. More results concerning such codes can be found in
\cite{stichtenoth:2009,niederreiter:2009}. We follow the notation of
\cite{stichtenoth:2009}.

The theory of algebraic geometry codes was introduced by Goppa in
his seminal work \cite{goppa:1977}. Let ${\mathbb F}_{q}$ be the
finite field with $q$ elements, where $q$ is a prime power and let
$F /{\mathbb F}_{q}$ be an algebraic function field of genus $g$. We
denote by ${\mathcal P}_{F}$ the set of places of $F /{\mathbb
F}_{q}$ and by ${\mathcal D}_{F}$ the (free) group of divisors of $F
/{\mathbb F}_{q}$. For each $ x \in F /{\mathbb F}_{q}$, the
principal divisor $(x)$ of $x$ is defined by
$(x):=\displaystyle\sum_{P}v_{P}(x) P$, where $v_{P}$ is the
discrete valuation corresponding to the place $P$. Let $A$ be a
divisor of $F /{\mathbb F}_{q}$. Then we define $l(A):= \dim
{\mathcal L}(A)$, where ${\mathcal L}(A)$ is the Riemann-Roch space
associated to $A$, given by ${\mathcal L}(A):= \{ x \in F \ | \
(x)\geq -A \} \cup \{ 0 \}$. Let ${\Omega}_{F}$ be the differential
space of $F /{\mathbb F}_{q}$, i.e., ${\Omega}_{F}:= \{ w \ | \ w \
is \ a \ Weil \ differential \ of \ F /{\mathbb F}_{q} \}$. For
every nonzero differential $w$, we denote its canonical divisor by
$(w):= \displaystyle\sum_{P}v_{P}(w) P$, where
$v_{P}(w):=v_{P}((w))$. A divisor $W$ is called canonical if $W=(w)$
for some $w \in {\Omega}_{F}$.

\begin{theorem}(Riemann-Roch)(Thm. 1.5.15 of \cite{stichtenoth:2009})
Let $W$ be a canonical divisor of $F /{\mathbb F}_{q}$. Then, for
each divisor $A \in {\mathcal D}_{F}$, the dimension of $ {\mathcal
L}(A)$ is given by $l(A) = \deg A +1 - g + \dim (W - A)$.
\end{theorem}

In what follows, we assume that $P_1 , \ldots , P_n$ are pairwise
distinct places of $F /{\mathbb F}_{q}$ of degree $1$, and $D = P_1
+ \ldots + P_n $ is a divisor. Let $G$ be a divisor of $F /{\mathbb
F}_{q}$ such that $ \operatorname{supp} G \cap \operatorname{supp} D
= \emptyset$. The geometric Goppa code $C_{\mathcal L}(D, G)$
associated with $D$ and $G$ is defined by $C_{\mathcal L}(D, G):= \{
(x(P_1), \ldots , x(P_n))| x \in {\mathcal L}(G)\} \subseteq
{\mathbb F}_{q}^{n}$.

\begin{theorem}(Thm. 2.2.2./Cor.2.2.3 of \cite{stichtenoth:2009})\label{h1} Under the hypotheses above,
$C_{\mathcal L}(D, G)$ is an $[n, k, d]_{q}$ code with $k =
l(G)-l(G-D)$ and $d \geq n - \deg G$. In addition, if $2g - 2 < \deg
G < n$, then one has $k = \deg G + 1 - g$.
\end{theorem}

Let $D$ and $G$ as above. We define the code $C_{\Omega}(D, G)
\subseteq {\mathbb F}_{q}^{n}$ by $ C_{\Omega}(D, G):= \{
({resp}_{P_1}(w), \ldots , {resp}_{P_n}(w))| w \in {\Omega}_{F}(G -
D)\}$, where ${resp}_{P_i} (w)$ denotes the residue of $w$ at $P_i$.

\begin{theorem}(Thm. 2.2.7. of \cite{stichtenoth:2009})\label{h2}
If $2g - 2 < \deg G < n$, then $C_{\Omega}(D, G)$ is an $[n, k^{*},
d^{*}]_{q}$ code, with $k^{*} = n + g - 1 - \deg G$ and $d^{*} \geq
\deg G -(2g - 2)$.
\end{theorem}

\begin{theorem}(Thm. 2.2.8. of \cite{stichtenoth:2009})\label{h3}
The codes $C_{\mathcal L}(D, G)$ and $C_{\Omega}(D, G)$ are
(Euclidean) dual of each other, i.e., $C_{\Omega}(D, G)= C_{\mathcal
L}(D, G)^{\perp}$.
\end{theorem}

\subsection{Quantum codes}\label{sec2.2}

In this subsection, we recall some necessary concepts concerning
quantum codes and also the CSS code construction. For more details
on quantum coding theory, we refer the reader to
\cite{nielsen:2000,calderbank:1998} in the case of the binary or
quaternary alphabet, and the paper \cite{ketkar:2006} in the general
case of nonbinary alphabets.

Recall that a $q$-ary quantum code ${\mathbb Q}$ of length $n$ is a
$K$-dimensional subspace of the $q^{n}$-dimensional Hilbert space
${({\mathbb C}^{q})}^{\otimes n}$, where $\otimes n$ denotes the
tensor product of vector spaces. If $K=q^{k}$ we write $[[n, k,
d]]_{q}$ to denote a $q$-ary quantum code of length $n$ and minimum
distance $d$. Let $[[n, k, d]]_{q}$ be a quantum code. The quantum
Singleton bound (QSB) asserts that $k + 2d \leq n + 2$. If the
equality holds then the code is called a \emph{maximum distance
separable} (MDS) code.

\begin{lemma}\cite{nielsen:2000,calderbank:1998,ketkar:2006}(CSS
construction) Let $C_1$ and $C_2$ denote two classical linear codes
with parameters ${[n, k_1, d_1]}_{q}$ and ${[n, k_2, d_2]}_{q}$,
respectively, such that $C_1\subset C_2$. Then there exists an
${[[n, k = k_2- k_1, d]]}_{q}$ quantum code, where $d = \min \{wt(c)
| c \in (C_2 \backslash C_1) \cup (\displaystyle C_{1}^{\perp}
\backslash \displaystyle C_{2}^{\perp}) \}$.
\end{lemma}

\section{The New Codes}\label{sec3}

This section is divided into three parts. The first subsection deals
with constructions of quantum $t$-point algebraic geometry codes. In
the second one, we construct AG codes where the divisor $G$ is a sum
of non-rational places and, in the third subsection, we construct
sequences of asymptotically good quantum AG codes.

\subsection{Construction I}\label{sec3.1}

In this section we present the contributions of this work. The first
result utilizes two $t$-point AG codes, where $t \geq 1$, in order
to derive quantum codes with good parameters.

\begin{theorem}(General $t$-point construction, $t \geq 1$)\label{main1}
Let $q$ be a prime power and $F /{\mathbb F}_{q}$ be an algebraic
function field of genus $g$, with $n+t$ pairwise distinct rational
places. Assume that $a_i, b_i$, $i= 1, \ldots , t$, are positive
integers such that $a_i\leq b_i$ for all $i$, and $2g - 2 <
\displaystyle\sum_{i=1}^{t}a_{i} < \displaystyle\sum_{i=1}^{t}b_{i}
< n$. Then there exists a quantum code with parameters ${[[n, k,
d]]}_{q}$, with $k =
\displaystyle\sum_{i=1}^{t}b_{i}-\displaystyle\sum_{i=1}^{t}a_{i}$
and $d \geq \min \left\{ n - \displaystyle\sum_{i=1}^{t}b_{i},
\displaystyle\sum_{i=1}^{t}a_{i}-(2g - 2) \right\}$.
\end{theorem}

\begin{proof}
Let $\{P_1 , P_2, \ldots ,P_n , P_{n+1}, \ldots , P_{n+t}\}$ be the
set of places of $F /{\mathbb F}_{q}$ of degree one. Let $D= P_1 +
\ldots + P_n$ be a divisor of $F /{\mathbb F}_{q}$. Assume also that
$G_1$ and $G_2$ are two divisors of $F /{\mathbb F}_{q}$ given,
respectively, by $G_1 = a_{1}P_{n+1} + \ldots + a_{t}P_{n+t}$ and
$G_2 = b_{1}P_{n+1} + \ldots + b_{t}P_{n+t}$, where $a_i \leq b_i $
for all $i = 1, \ldots, t$ and $2g - 2 <
\displaystyle\sum_{i=1}^{t}a_{i} < \displaystyle\sum_{i=1}^{t}b_{i}
< n$. From construction, $ \operatorname{supp} G_1 \cap
\operatorname{supp} D = \emptyset$ and $ \operatorname{supp} G_2
\cap \operatorname{supp} D = \emptyset$. Since $G_1 \leq G_2$, one
has ${\mathcal L}(G_1) \subset {\mathcal L}(G_2)$; so $C_{\mathcal
L}(D, G_1) \subset C_{\mathcal L}(D, G_2 )$. From Theorem~\ref{h1},
the code $C_1 := C_{\mathcal L}(D, G_1)$ has parameters $[n, k_1 ,
d_1]_{q}$, where $d_1 \geq n - \displaystyle\sum_{i=1}^{t}a_{i}$ and
$k_1 = \displaystyle\sum_{i=1}^{t}a_{i} - g +1$; the code $C_2:=
C_{\mathcal L}(D, G_2)$ has parameters $[n, k_2 , d_2]_{q}$, where
$d_2 \geq n - \displaystyle\sum_{i=1}^{t}b_{i}$ and $k_2 =
\displaystyle\sum_{i=1}^{t}b_{i} - g +1$. On the other hand, from
Theorems~\ref{h2} and \ref{h3}, the dual code $C_{1}^{\perp}=
C_{\Omega}(D, G_1)$ of $C_1$ has parameters $[n, k_{1}^{\perp},
d_{1}^{\perp}]_{q}$, where $d_{1}^{\perp} \geq
\displaystyle\sum_{i=1}^{t}a_{i} -(2g - 2)$ and $k_{1}^{\perp}= n +
g - 1 - \displaystyle\sum_{i=1}^{t}a_{i}$; the dual code
$C_{2}^{\perp}= C_{\Omega}(D, G_2)$ of $C_2$ has parameters $[n,
k_{2}^{\perp}, d_{2}^{\perp}]_{q}$, with $d_{2}^{\perp} \geq
\displaystyle\sum_{i=1}^{t}b_{i} -(2g - 2)$ and $k_{2}^{\perp}= n +
g - 1 - \displaystyle\sum_{i=1}^{t}b_{i}$.

Applying the CSS construction to the codes $C_1$ and $C_2$, we
obtain a quantum code with parameters ${[[n, k, d]]}_{q}$, with
$k=k_2 - k_1 = (\displaystyle\sum_{i=1}^{t}b_{i} - g +1)-
(\displaystyle\sum_{i=1}^{t}a_{i} - g +1)=
\displaystyle\sum_{i=1}^{t}b_{i} - \displaystyle\sum_{i=1}^{t}a_{i}$
and $d \geq \min \{d_2 , d_{1}^{\perp}\}$, where $d_2 \geq n -
\displaystyle\sum_{i=1}^{t}b_{i}$ and $d_{1}^{\perp} \geq
\displaystyle\sum_{i=1}^{t}a_{i} -(2g - 2)$. The proof is complete.
\end{proof}

\begin{remark}
In \cite{ling:2001,walker:2008}, the authors utilized one-point AG
codes to construct good/(asymptotically good) quantum codes. In
\cite{chen:2001}, the author applied two-point AG codes to derive
good/(asymptotically good) quantum codes. Note that, in this
context, Theorem~\ref{main1} is a natural generalization of the
one-point as well as two-point AG code construction to the $t$-point
($t \geq 1$) AG code construction.
\end{remark}

\begin{corollary}\label{main2}(One-Point codes)
There exists a quantum code with parameters $[[q(1 + (q-1)m),$ $b-a,
d]{]}_{q^{2}}$, where $(q-1)(m-1)-2 < a < b < q(1 + (q-1)m)$,
$m|(q+1)$ and $d \geq \min \{q(1 + (q-1)m) - b, a -(q-1)(m-1) +2\}$.
\end{corollary}

\begin{proof}
Let $F = {\mathbb F}_{q^{2}}(x, y)$, where $y^{q}+y = x^{m}$ and
$m|(q+1)$. It is known that the genus of $F$ is equal to $g =
(q-1)(m-1)/2$, and the number of places of degree one is $N=1+ q(1 +
(q-1)m)$ (see Example 6.4.2. of \cite{stichtenoth:2009}). Let $\{P_1
, P_2, \ldots ,P_n , P_{n+1}, \ldots , P_{N}\}$ be these pairwise
distinct places. Without loss of generality, choose the ${\mathbb
F}_{q^{2}}$-rational point $P_{N}$. Let $D= P_1 + \ldots + P_{N-1}$
be a divisor and let $G_1= a P_N$ and $G_2 = b P_N$ other two
divisors such that $ \operatorname{supp} G_1 \cap
\operatorname{supp} D = \emptyset$ and $ \operatorname{supp} G_2
\cap \operatorname{supp} D = \emptyset$, where $(q-1)(m-1)-2 < a < b
< q(1 + (q-1)m)$. From Theorem~\ref{main1}, there exists a quantum
code with parameters ${[[q(1 + (q-1)m), b-a, d]]}_{q^{2}}$, where $d
\geq \min \{q(1 + (q-1)m) - b, a -(q-1)(m-1) + 2\}$. The proof is
complete.
\end{proof}

\begin{remark}
Note that the Hermitian curve defined as $y^{q}+y= x^{q+1}$, over
${\mathbb F}_{q^{2}}$, is a particular case of the curve $y^{q}+y =
x^{m}$, considered in the proof of Corollary~\ref{main2}.
\end{remark}

\begin{corollary}\label{main3}(Two-Point codes)
There exists a quantum code with parameters $[[q(1 + (q-1)m)-1, b_1
+ b_2- a_1 - a_2, d]{]}_{q^{2}}$, where $a_i \leq b_i$ for $i=1, 2$,
$(q-1)(m-1)-2 < a_1 + a_2 < b_1 + b_2 < q[1 + (q-1)m]-1$, $m|(q+1)$
and $d \geq \min \{q[1 + (q-1)m] - b_1 - b_2 -1, a_1 + a_2
-(q-1)(m-1) +2\}$.
\end{corollary}

\begin{proof}
Let $D= P_1 + \ldots + P_{N-2}$ be a divisor and let $G_1= a_1
P_{N-2} + a_2 P_{N-1}$ and $G_2 = b_1 P_{N-2} + b_2 P_{N-1}$ be
other two divisors with $ \operatorname{supp} G_1 \cap
\operatorname{supp} D = \emptyset$ and $ \operatorname{supp} G_2
\cap \operatorname{supp} D = \emptyset$, where $(q-1)(m-1)-2 < a_1 +
a_2 < b_1 + b_2 < q(1 + (q-1)m)-1$. From Theorem~\ref{main1}, there
exists a quantum code with parameters ${[[q(1 + (q-1)m)-1, b_1 +
b_2-a_1 - a_2, d]]}_{q^{2}}$, where $d \geq \min \{q(1 + (q-1)m)-1 -
b_1 - b_2, a_1 + a_2 -(q-1)(m-1) + 2\}$.
\end{proof}

\begin{corollary}\label{main4}($t$-Point codes, $t\geq 2$)
There exists a quantum code with parameters $[[q(1 + (q-1)m)-t+1,
b_1 + \ldots + b_t - (a_1 + \ldots + a_t), d]{]}_{q^{2}}$, where
$a_i \leq b_i$ for $i=1, \ldots t$, $(q-1)(m-1)-2 < a_1 + \ldots +
a_t < b_1 + \ldots + b_t < q(1 + (q-1)m)-t + 1$, $m|(q+1)$ and $d
\geq \min \{q(1 + (q-1)m) -(b_1 + \ldots + b_t)-t+1, a_1 + \ldots +
a_t -(q-1)(m-1) +2\}$.
\end{corollary}
\begin{proof}
Similar to that of Corollary~\ref{main3}.
\end{proof}

\subsection{Construction II}\label{sec3.2}

In this section we deal with constructions of quantum codes derived
from AG codes whose divisors are multiples of a non rational divisor
$G$. The first result is given in the following.

\begin{theorem}(General construction)\label{main5}
Let $q$ be a prime power and let $F /{\mathbb F}_{q}$ be an
algebraic function field of genus $g$, with $n$ pairwise distinct
rational places $P_i$, $i=1, \ldots , n$. Assume that there exist
pairwise distinct places $Q_1 , \ldots , Q_t$ of $F /{\mathbb
F}_{q}$, of degree ${\alpha}_{i} \geq 2$, respectively, $i=1, \ldots
, t$, where $t\geq 1$. Let $G_1
=\displaystyle\sum_{i=1}^{t}a_{i}Q_{i}$ and $G_2
=\displaystyle\sum_{i=1}^{t}b_{i}Q_{i}$, where $a_i \leq b_i$, for
all $i= 1, \ldots , t$, and $2g - 2 < a_{1}{\alpha}_{1}+\ldots +
a_{t}{\alpha}_{t} < b_{1}{\alpha}_{1}+\ldots + b_{t}{\alpha}_{t} <
n$. Let $D = P_1 + \ldots + P_n $ be a divisor of $F /{\mathbb
F}_{q}$, and consider that $ \operatorname{supp} G_1 \cap
\operatorname{supp} D = \emptyset$ and  $ \operatorname{supp} G_2
\cap \operatorname{supp} D = \emptyset$. Then there exists a quantum
code with parameters ${[[n, k, d]]}_{q}$, where $k =
(b_{1}-a_{1}){\alpha}_{1}+\ldots + (b_{t}- a_{t}){\alpha}_{t}$ and
$d \geq \min \{n - (b_{1}{\alpha}_{1}+\ldots + b_{t}{\alpha}_{t}),
(a_{1}{\alpha}_{1}+\ldots + a_{t}{\alpha}_{t}) -(2g - 2)\}$.
\end{theorem}
\begin{proof}
The proof is similar to that of Theorem~\ref{main1}.
\end{proof}

\begin{corollary}\label{main6}
Let $q$ be a prime power and let $F /{\mathbb F}_{q}$ be a
hyperelliptic function field of genus $g\geq 2$, with $n$ pairwise
distinct rational places. Then there exists a quantum code with
parameters ${[[n, 2(t_2 - t_1), d]]}_{q}$, where $t_1 , t_2$ are
positive integers such that $2g - 2 < t_1 < t_2 < n$ and $d\geq \min
\{ n - 2t_2 , 2t_1 - 2g +2 \}$.
\end{corollary}

\begin{proof}
Since $F$ is a hyperelliptic function field then there exists a
place $G$ of degree two (see Lemma 6.2.2.(a) of
\cite{stichtenoth:2009}). Let $D= P_1 + \ldots + P_{n}$ be a
divisor, where $P_i$ are all rational points of $F$. Let $G_2 = t_2
G$ and $G_1 = t_1 G$, where $2g-2 < 2t_1 < 2t_2 < n$. We know that $
\operatorname{supp} G_1 \cap \operatorname{supp} D = \emptyset$, $
\operatorname{supp} G_2 \cap \operatorname{supp} D = \emptyset$ and
$C_{\mathcal L}(D, G_1) \subset C_{\mathcal L}(D, G_2 )$. From
Theorem~\ref{h1}, the code $C_1 := C_{\mathcal L}(D, G_1)$ has
parameters $[n, k_1 , d_1]_{q}$, where $d_1 \geq n - 2t_1$ and $k_1
= 2t_1 - g +1$. The code $C_2:= C_{\mathcal L}(D, G_2)$ has
parameters $[n, k_2 , d_2]_{q}$, where $d_2 \geq n - 2t_2$ and $k_2
= 2t_2 - g +1$.

From Theorems~\ref{h2} and \ref{h3}, the dual code $C_{1}^{\perp}=
C_{\Omega}(D, G_1)$ of $C_1$ has parameters $[n, k_{1}^{\perp},
d_{1}^{\perp}]_{q}$, where $d_{1}^{\perp} \geq 2t_1 -(2g - 2)$ and
$k_{1}^{\perp}= n + g - 1 - 2t_1$. Similarly, the dual code
$C_{2}^{\perp}= C_{\Omega}(D, G_2)$ of $C_2$ has parameters $[n,
k_{2}^{\perp}, d_{2}^{\perp}]_{q}$, where $d_{2}^{\perp} \geq 2t_2
-(2g - 2)$ and $k_{2}^{\perp}= n + g - 1 - 2t_2$.

Applying the CSS construction to the codes $C_1$ and $C_2$, we
obtain an ${[[n, k, d]]}_{q}$ quantum code, with $k=k_2 - k_1 =
(2t_2 - g +1)- (2t_1 - g +1) = 2(t_2 - t_1)$ and $d \geq \min \{d_2
, d_{1}^{\perp}\}$, where $d_2 \geq n - 2t_2$ and $d_{1}^{\perp}
\geq 2t_1 -(2g - 2)$. Thus, the result follows.
\end{proof}

\begin{corollary}\label{main7}
There exists a quantum code with parameters ${[[46, 2(t_2 - t_1),
d]]}_{25}$, where $t_1 , t_2$ are positive integers such that $1
<t_1 < t_2 < 23$ and $d\geq \min \{ 46 - 2t_2 , 2t_1 - 2 \}$.
\end{corollary}

\begin{proof}
Let us consider the function field $F = {\mathbb F}_{q^{2}}(x, y)$
with $y^{q}+y = x^{m}$ and $m|(q+1)$; take $m=2$ and $q=5$. As the
genus of $F$ is $g=2$, then $F$ is a hyperelliptic function field
(see Lemma 6.2.2.(b) of \cite{stichtenoth:2009}), and the result
follows from Corollary~\ref{main6}.
\end{proof}

\subsection{Construction III}\label{sec3.3}

In this subsection, we propose constructions of sequences of
asymptotically good quantum codes derived from AG codes.

Recall that a tower of function fields (see Def. 1.3 of
\cite{henning:1996}) over ${\mathbb F}_{q}$ is a sequence ${\mathcal
T} = (F_1 , F_2, \ldots )$ of function fields $F_{i}/{\mathbb
F}_{q}$ with the following properties:
\begin{itemize}
\item [ (1)] $F_1\subseteq F_{2} \subseteq F_3 \cdots$.

\item [ (2)] For each $n\geq 1$, the extension $F_{n+1}/F_{n}$ is separable of degree
$[F_{n+1} : F_{n}] > 1$.

\item [ (3)] $g(F_{j}) > 1$, for some $j > 1$.
\end{itemize}

By the Hurwitz genus formula, the condition $(3)$ implies that
$g(F_{n})\rightarrow \infty$ for $n\rightarrow \infty$. The tower is
said to be asymptotically good if $\lambda ({\mathcal T})=
{\limsup}_{i\rightarrow\infty}$ $N(F_i)/g(F_i) > 0$, where $N(F_i)$
and $g(F_i)$ denote the number of ${\mathbb F}_{q}$-rational points
and the genus of $F_i$, respectively. In the case of tower of
function field one can replace ${\limsup}_{i\rightarrow\infty}
N(F_i)/g(F_i)$ by ${\lim}_{i\rightarrow\infty} N(F_i)/g(F_i)$,
because the sequence ${(N(F_i)/g(F_i))}_{i\geq 1}$ is convergent. We
say that the tower ${\mathcal T}$ (over ${\mathbb F}_{q}$) attains
the Drinfeld-Vladut bound if $\lambda ({\mathcal T}) =
{\limsup}_{i\rightarrow\infty} N(F_i)/g(F_i)=\sqrt{q}-1$. To
simplify the notation we put $N(F_i)=N_i$ and $g(F_i)=g_i$.

Let ${({\mathcal Q}_{i})}_{i\geq 1}$ be a sequence of quantum codes
over ${\mathbb F}_{q}$ with parameters $[[n_i , k_i ,$
$d_i]{]}_{q}$, respectively. We say that ${({\mathcal
Q}_{i})}_{i\geq 1}$ is asymptotically good if
${\limsup}_{i\rightarrow\infty} k_i/n_i$ $> 0$ and
${\limsup}_{i\rightarrow\infty} d_i/n_i > 0$. The next result shows
how to construct asymptotically good quantum codes derived from
(classical) two-point AG codes.

\begin{theorem}(Two-point asymptotically good codes)\label{main8}
Assume that the tower ${\mathcal T} = (F_1 , F_2, \ldots )$ of
function fields over ${\mathbb F}_{q}$ attains the Drinfeld-Vladut
bound. Then there exists a sequence ${({\mathcal Q}_{i})}_{i\geq 1}$
of asymptotically good quantum codes, over ${\mathbb F}_{q}$,
derived from classical two-point AG codes.
\end{theorem}

\begin{proof}
For each $F_i$, let us consider the set of rational places $P_1(i),
\ldots, P_{N_{i} -2}(i),$ $P_{N_{i}-1}(i), P_{N_{i}}(i)$ of $F_i$.
We set the divisors $D(i)=P_1(i)+ \ldots + P_{N_{i} -2}(i)$, $G_1
(i) = a_1(i)P_{N_{i} -1}(i) + a_2(i)P_{N_{i}}(i)$ and $G_2 (i) =
b_1(i)P_{N_{i} -1}(i) + b_2(i)P_{N_{i}}(i)$, where $a_1(i)\leq
b_1(i)$ and $a_2(i)\leq b_2(i)$, with $2g_{i} - 2 < a_1(i) + a_2(i)<
b_1(i) + b_2(i) < N_{i}-2$. Let $C_{1}(i) := C_{\mathcal L}(i)[D(i),
G_1(i)]$ and $C_{2}(i) := C_{\mathcal L}(i)[D(i), G_2(i)]$ be the
two-point AG codes, over ${\mathbb F}_{q}$, corresponding to $G_1
(i)$ and $G_2 (i)$, respectively; thus $C_{1}(i)\subset C_{2}(i)$.
The code $C_1 (i)$ has parameters $[N_{i}-2, a_1(i) + a_2(i)-g_i +1
, d_1 (i)]_{q}$, where $d_{1}(i) \geq N_{i}-2 - (a_1(i) + a_2(i))$;
the code $C_2(i)$ has parameters $[N_{i}-2, b_1(i) + b_2(i)-g_i +1,
d_2 (i)]_{q}$, where $d_{2}(i) \geq N_{i}-2 - (b_1(i) + b_2(i))$.
The corresponding CSS code has parameters ${[[N_{i}-2, K_i = b_1(i)
+ b_2(i)-( a_1(i) + a_2(i)), D_i]]}_{q}$, where $D_i \geq \min
\{N_{i}-2 - (b_1(i) + b_2(i)), a_{1}(i) + a_{2}(i) -(2g_i - 2)\}$.
We know that the $K_i$'s assume all the values from $1$ to $N_i -
2g_i -2$, i.e., $0 < K_i \leq N_i - 2g_i -2$. For any such $K_i$ we
set $b_1 (i) + b_2 (i) = \lfloor (N_i + 2g_i + K_i - 4)/2 \rfloor$;
thus it follows that $N_{i}-2 - (b_1(i) + b_2(i))\geq a_{1}(i) +
a_{2}(i) -(2g_i - 2)$, where $a_{1}(i) + a_{2}(i) -(2g_i - 2)\geq
(N_i - K_i -2g_i - 1)/2$. The sequence of positive integers
${(K_{i})}_{i\geq 1}$ satisfies $0 < {\limsup}_{i\rightarrow\infty}
\frac{K_i}{N_i-2}\leq {\limsup}_{i\rightarrow\infty}N_{i}/(N_{i}-2)
- {\limsup}_{i\rightarrow\infty} 2g_i/(N_{i}-2)+
{\limsup}_{i\rightarrow\infty} -2/(N_{i}-2)= 1 -2/(\sqrt{q} -1)$,
where in the last equality we use the fact that
${\limsup}_{i\rightarrow\infty} N_{i}/g_{i} = \sqrt{q} -1$. For each
$0 < c < 1 -2/(\sqrt{q} -1)$, we can choose convenient values for
$K_i$ such that ${\lim}_{i\rightarrow\infty} K_{i}/N_{i}= c$. Thus,
${\limsup}_{i\rightarrow\infty} K_{i}/(N_{i}-2)=c >0$. Moreover one
has ${\limsup}_{i\rightarrow\infty} (N_i - K_i -2g_i - 1)/2(N_{i} -
2)= 1/2\left[ 1 -2/(\sqrt{q} -1) -c\right] > 0$. Therefore, there
exists a sequence ${({\mathcal Q}_{i})}_{i\geq 1}$ of asymptotically
good quantum codes over ${\mathbb F}_{q}$. The proof is complete.
\end{proof}

\begin{remark}
Since several works available in the literature already presented
constructions of asymptotically good quantum codes derived from
one-point AG codes (see \cite{ashikhmin:2001,chen:2001,ling:2001}),
we do not present such constructions in this paper.
\end{remark}

\begin{theorem}($t$-point asymptotically good codes)\label{main9}
Assume that the tower ${\mathcal T} = (F_1 , F_2, \ldots )$ of
function fields over ${\mathbb F}_{q}$ attains the Drinfeld-Vladut
bound. Then there exists a sequence ${({\mathcal Q}_{i})}_{i\geq 1}$
of asymptotically good quantum codes, over ${\mathbb F}_{q}$,
derived from classical $t$-point AG codes.
\end{theorem}

\begin{proof}
We adopt the same notation as in the proof of Theorem~\ref{main8}.
For each $F_i$, let us consider the set of rational places $P_1(i),
\ldots, P_{n_{i}}(i),$ $P_{n_{i}+1}(i), \ldots, P_{n_{i}+t}(i)$ of
$F_i$, where $N_i = n_i + t$. Set $D(i)=P_1(i)+ \ldots +
P_{n_{i}}(i)$, $G_1 (i) = a_1(i)P_{n_{i}+1}(i) +\ldots +
a_t(i)P_{n_{i}+t}(i)$ and $G_2 (i) = b_1(i)P_{n_{i}+1}(i) +\ldots +
b_t(i)P_{n_{i}+t}(i)$, where $a_j(i)\leq b_j(i)$ for all $j=1,
\ldots, t$, with $2g_{i} - 2 < \displaystyle\sum_{j=1}^{t}a_j(i) <
\displaystyle\sum_{j=1}^{t}b_j(i) < N_{i}-t$. Let us consider the
$t$-point AG codes $C_{1}(i) := C_{\mathcal L}(i)[D(i), G_1(i)]$ and
$C_{2}(i) := C_{\mathcal L}(i)[D(i), G_2(i)]$. It follows that
$C_{1}(i)\subset C_{2}(i)$, and $C_1 (i)$ has parameters $\left[
N_{i}-t, \displaystyle\sum_{j=1}^{t}a_j(i)-g_i +1 , d_1
(i)\right]_{q}$, where $d_{1}(i) \geq N_{i}-t -
\displaystyle\sum_{j=1}^{t}a_j(i)$. Moreover, $C_2(i)$ has
parameters $\left[ N_{i}-t, \displaystyle\sum_{j=1}^{t}b_j(i)-g_i
+1, d_2 (i)\right]_{q}$, where $d_{2}(i) \geq N_{i}-t -
\displaystyle\sum_{j=1}^{t}b_j(i)$.

Setting $\displaystyle\sum_{j=1}^{t}b_j(i) = \lfloor (N_i + 2g_i +
K_i - t- 2)/2 \rfloor$ and proceeding similar as in the proof of
Theorem~\ref{main8}, the result follows.
\end{proof}

Let $q$ be a prime power and $C = {[n, k, d]}_{q^{m}}$ be a linear
code over ${\mathbb F}_{q^{m}}$. Let $\beta$ be a basis of ${\mathbb
F}_{q^{m}}$ over ${\mathbb F}_{q}$ and assume also that
${\beta}^{\perp}$ is a dual basis of $\beta$. Let $C^{\perp}$ be the
Euclidean dual of $C$. Then one has ${[\beta (C)]}^{\perp}=
{\beta}^{\perp}(C^{\perp})$ (see \cite{grassl:1999,laguardia:2012}).

\begin{theorem}\label{main10}
For any prime $p$, there exists a sequence ${({\mathcal
Q}_{i})}_{i\geq 1}$ of asymptotically good quantum codes over
${\mathbb F}_{p}$.
\end{theorem}

\begin{proof}
Let $q^2 = p^{2r}$, $p$ prime. Let us consider the tower of function
fields ${\mathcal T} = (F_1 , F_2, \ldots )$ over ${\mathbb
F}_{q^{2}}$, shown in \cite{henning:1996}, defined by $F_t =
{\mathbb F}_{q^{2}} (x_1, \ldots , x_t)$, where $x_{i+1}^{q}+
x_{i+1}= x_{i}^{q}/(x_{i}^{q-1}+1)$, for $ i = 1, \ldots t-1$. This
tower attains the Drinfeld-Vladut bound. We next expand the codes
$C_1 (i)$ and $C_2(i)$, shown in the proof of Theorem~\ref{main8},
with respect to some basis $\beta$ of ${\mathbb F}_{q^{2}}$ over
${\mathbb F}_{p}$. Thus, we obtain codes $\beta (C_{1}(i))$ and
$\beta (C_{2}(i))$, both over ${\mathbb F}_{p}$, with parameters
$[2r(N_{i}-2), 2r(a_1(i) + a_2(i)-g_i +1) , \geq d_{1}^{*}
(i)]_{p}$, where $d_{1}^{*}(i)\geq d_{1}(i) \geq N_{i}-2 - (a_1(i) +
a_2(i))$, and $[2r(N_{i}-2), 2r(b_1(i) + b_2(i)-g_i +1),
d_{2}^{*}(i)]_{p}$, where $d_{2}^{*}(i) \geq d_{2}(i) \geq N_{i}-2 -
(b_1(i) + b_2(i))$, respectively. Because $\beta (C_{1}(i))\subset
\beta (C_{2}(i))$, we apply the CSS construction to these codes,
obtaining therefore an ${[[2r(N_{i}-2),  2rK_i, D_i]]}_{p}$ quantum
code, where $D_i \geq \min \{N_{i}-2 - (b_1(i) + b_2(i)), a_{1}(i) +
a_{2}(i) -(2g_i - 2)\}$ (note that since ${[\beta
(C_{1}(i))]}^{\perp}= {\beta}^{\perp}{(C_{1}(i))}^{\perp}$, then the
minimum distance of ${[\beta (C_{1}(i))]}^{\perp}$ is at least
$a_{1}(i) + a_{2}(i) -(2g_i - 2)$). Proceeding similarly as in the
proof of Theorem~\ref{main8}, we get $N_{i}-2 - (b_1(i) +
b_2(i))\geq a_{1}(i) + a_{2}(i) -(2g_i - 2) \geq (N_i - K_i -2g_i +
1)/2$. Consequently, one has ${\limsup}_{i\rightarrow\infty} 2r
K_{i}/2r(N_{i}-2) >0$ and ${\limsup}_{i\rightarrow\infty} (N_i - K_i
-2g_i + 1)/4r(N_i - 2)= 1/4r \left[ 1 -2/(p^r -1) -c \right] > 0$,
as desired.
\end{proof}

\begin{remark}
Although the proofs of Theorems~\ref{main8}~and~\ref{main10} are
similar to the proofs of the corresponding results shown in
references \cite{ling:2001,walker:2008}, in the present paper we
utilize $t$-point ($t\geq 2$) AG codes, whereas in such references,
the authors utilized only one-point AG codes to perform their
constructions. Another difference is that in
\cite{ling:2001,walker:2008}, the authors utilized the technique of
code concatenation to obtain (quantum) codes over prime fields;
here, we utilize the technique of code expansion.
\end{remark}

\section{Examples and Code Comparison}\label{sec4}

In Tables~\ref{table1}, \ref{table2} and \ref{table3}, we exhibit
some new quantum codes derived from
Corollaries~\ref{main2},~\ref{main3}~and~\ref{main7}, respectively.
In these tables, $q$ is a prime power and $a, b, a_1, a_2, b_1 ,$
$b_2 , t_1 , t_2 , m$ are positive integers satisfying some
conditions. More precisely: in Table~\ref{table1}, we consider that
$(q-1)(m-1)-2 < a < b$, $b < q(1 + (q-1)m)$ and $m|(q+1)$; in
Table~\ref{table2}, we assume that $a_i \leq b_i$ $i=1, 2$,
$(q-1)(m-1)-2 < a_1 + a_2 < b_1 + b_2$, $b_1 + b_2 < q[1 +
(q-1)m]-1$ and $m|(q+1)$; in Table~\ref{table3}, we suppose that $1
< t_1 < t_2 < 23$.

\begin{table}[!hpt]
\begin{center}
\caption{New quantum codes\label{table1}}
\begin{tabular}{|c|c|c|c|c|}
\hline New codes from Corollary~\ref{main2} & $q$ & $m$ & $a$ &
$b$\\
\hline ${[[27, 17, d\geq 3]]}_{9}$ & $3$ & $4$ & $7$ & $24$\\
\hline ${[[27, 15, d\geq 4]]}_{9}$ & $3$ & $4$ & $8$ & $23$\\
\hline ${[[27, 13, d\geq 5]]}_{9}$ & $3$ & $4$ & $9$ & $22$\\
\hline ${[[27, 11, d\geq 6]]}_{9}$ & $3$ & $4$ & $10$ & $21$\\
\hline ${[[27, 9, d\geq 7]]}_{9}$ & $3$ & $4$ & $11$ & $20$\\
\hline ${[[27, 7, d\geq 8]]}_{9}$ & $3$ & $4$ & $12$ & $19$\\
\hline ${[[27, 5, d\geq 9]]}_{9}$ & $3$ & $4$ & $13$ & $18$\\
\hline ${[[27, 3, d\geq 10]]}_{9}$ & $3$ & $4$ & $14$ & $17$\\
\hline ${[[27, 1, d\geq 11]]}_{9}$ & $3$ & $4$ & $15$ & $16$\\
\hline ${[[64, 48, d\geq 3]]}_{16}$ & $4$ & $5$ & $13$ & $61$\\
\hline ${[[64, 46, d\geq 4]]}_{16}$ & $4$ & $5$ & $14$ & $60$\\
\hline ${[[64, 44, d\geq 5]]}_{16}$ & $4$ & $5$ & $15$ & $59$\\
\hline ${[[64, 24, d\geq 15]]}_{16}$ & $4$ & $5$ & $25$ & $49$\\
\hline ${[[64, 4, d\geq 25]]}_{16}$ & $4$ & $5$ & $35$ & $39$\\
\hline ${[[64, 2, d\geq 26]]}_{16}$ & $4$ & $5$ & $36$ & $38$\\
\hline ${[[65, 53, d\geq 3]]}_{25}$ & $5$ & $3$ & $9$ & $62$\\
\hline ${[[65, 51, d\geq 4]]}_{25}$ & $5$ & $3$ & $10$ & $61$\\
\hline ${[[65, 49, d\geq 5]]}_{25}$ & $5$ & $3$ & $11$ & $60$\\
\hline ${[[65, 9, d\geq 25]]}_{25}$ & $5$ & $3$ & $31$ & $40$\\
\hline ${[[175, 153, d\geq 3]]}_{49}$ & $7$ & $4$ & $19$ & $172$\\
\hline ${[[175, 151, d\geq 4]]}_{49}$ & $7$ & $4$ & $20$ & $171$\\
\hline ${[[175, 149, d\geq 5]]}_{49}$ & $7$ & $4$ & $21$ & $170$\\
\hline ${[[175, 109, d\geq 25]]}_{49}$ & $7$ & $4$ & $41$ & $150$\\
\hline ${[[175, 31, d\geq 64]]}_{49}$ & $7$ & $4$ & $80$ & $111$\\
\hline ${[[175, 1, d\geq 79]]}_{49}$ & $7$ & $4$ & $95$ & $96$\\
\hline
\end{tabular}
\end{center}
\end{table}

\begin{table}[!hpt]
\begin{center}
\caption{New quantum codes\label{table2}}
\begin{tabular}{|c|c|c|c|c|c|c|}
\hline New codes from Corollary~\ref{main3} & $q$ & $m$ & $a_1$ &
$a_2$ & $b_1$ & $b_2$\\
\hline ${[[26, 16, d\geq 3]]}_{9}$ & $3$ & $4$ & $3$ & $4$ & $7$ & $16$\\
\hline ${[[26, 14, d\geq 4]]}_{9}$ & $3$ & $4$ & $3$ & $5$ & $7$ & $15$\\
\hline ${[[26, 12, d\geq 5]]}_{9}$ & $3$ & $4$ & $3$ & $6$ & $7$ & $14$\\
\hline ${[[26, 4, d\geq 9]]}_{9}$ & $3$ & $4$ & $3$ & $10$ & $7$ & $10$\\
\hline ${[[26, 2, d\geq 10]]}_{9}$ & $3$ & $4$ & $4$ & $10$ & $6$ & $10$\\
\hline
\end{tabular}
\end{center}
\end{table}

\begin{table}[!hpt]
\begin{center}
\caption{New quantum codes\label{table3}}
\begin{tabular}{|c|c|c|c|c|}
\hline New codes from Corollary~\ref{main7} & $q$ & $m$ & $t_1$ &
$t_2$\\
\hline ${[[46, 36, d\geq 4]]}_{25}$ & $5$ & $2$ & $3$ & $21$\\
\hline ${[[46, 32, d\geq 6]]}_{25}$ & $5$ & $2$ & $4$ & $20$\\
\hline ${[[46, 28, d\geq 8]]}_{25}$ & $5$ & $2$ & $5$ & $19$ \\
\hline ${[[46, 4, d\geq 20]]}_{25}$ & $5$ & $2$ & $11$ & $13$\\
\hline
\end{tabular}
\end{center}
\end{table}

Recall that the parameters of an $Q:=[[n, k, d]]_{q}$ quantum code
satisfy the inequality $k + 2d \leq n + 2$. This inequality is
called \emph{quantum Singleton bound} (QSB). The \emph{Singleton
defect} $(SD_{Q})$ of a code is defined as $SD_{Q}= n+2 -k-2d$. In
this paper, we measure the performance of the code by means of the
Singleton defect. We adopt this method because, for large alphabets,
it is difficult to find codes over them: ``... for large $q$, it is
difficult to find explicit known codes to compare with ours since
there are no suitable tables for reference" (see page 3 of
\cite{jin:2014}).

The new ${[[26, 16, d\geq 3]]}_{9}$ code is better than the ${[[26,
14, 3]]}_{9}$ code shown in Ref.~\cite{edel}, because the Singleton
defect of the new $Q_1 :={[[26, 16, d\geq 3]]}_{9}$ code is
$SD_{Q_1}\leq 6$, whereas the Singleton defect of the $Q_2 :={[[26,
14, 3]]}_{9}$ code is $SD_{Q_2}=8$. The new ${[[26, 14, d\geq
4]]}_{9}$ code with Singleton defect at most $8$ is better than the
${[[26, 4, 4]]}_{9}$ code shown in Ref.~\cite{edel}, which has
Singleton defect $16$. The new quantum codes of length $46$ have
Singleton defect at most $4$. Moreover, all new quantum codes of
lengths $n=26$ and $n=27$, exhibited in Table~\ref{table1}, have
Singleton defect at most $6$.

Note that the new ${[[27, 3, d\geq 10]]}_{9}$, ${[[27, 5, d\geq
9]]}_{9}$, ${[[65, 9, d\geq 25]]}_{25}$ and ${[[175, 31, d\geq
64]]}_{49}$ codes have large minimum distances when compared to
their code lengths.

The quantum codes shown in \cite{jin:2014} were constructed over the
field ${\mathbb F}_{q^{2}}$, where $q$ is a power of $2$, whereas in
this paper, we construct quantum codes over ${\mathbb F}_{q}$ for
all prime power $q$. Great part of the codes available in
\cite{jin:2014} were constructed over ${\mathbb F}_{8}$; this fact
does not allow us to compare our codes with the ones shown in
\cite{jin:2014}.

The quantum codes exhibited in \cite{munuera:2016} were constructed
over the fields ${\mathbb F}_{2}$, ${\mathbb F}_{3}$, ${\mathbb
F}_{4}$, ${\mathbb F}_{5}$, ${\mathbb F}_{8}$, ${\mathbb F}_{9}$. In
the present paper, we give examples of quantum codes constructed
over ${\mathbb F}_{9}$, ${\mathbb F}_{16}$, ${\mathbb F}_{25}$,
${\mathbb F}_{49}$. The codes over ${\mathbb F}_{9}$ constructed in
\cite{munuera:2016} are ${[[15, 13, 2]]}_{9}$, ${[[15, 7, 4]]}_{9}$,
${[[15, 5, 5]]}_{9}$, ${[[15, 1, 7]]}_{9}$, ${[[243, 241, 2]]}_{9}$,
${[[243, 219, 6]]}_{9}$, ${[[243, 213, 9]]}_{9}$.

The new codes over ${\mathbb F}_{9}$ shown in Tables~\ref{table1}
and \ref{table2} are ${[[26, 16, d\geq 3]]}_{9}$, $[[26, 14,$ $d\geq
4]]_{9}$, $[[26, 12,$ $d\geq 5]]_{9}$, ${[[26, 4, d\geq 9]]}_{9}$,
${[[26, 2, d\geq 10]]}_{9}$, ${[[27, 17, d\geq 3]]}_{9}$, $[[27,
15,$ $d\geq 4]]_{9}$, ${[[27, 13, d\geq 5]]}_{9}$, ${[[27, 11, d\geq
6]]}_{9}$, ${[[27, 9, d\geq 7]]}_{9}$, $[[27, 7,$ $d\geq 8]]_{9}$,
$[[27, 5,$ $d\geq 9]]_{9}$, ${[[27, 3, d\geq 10]]}_{9}$ and ${[[27,
1, d\geq 11]]}_{9}$. Since the parameters among these codes are
different, we do not perform the comparison.

\section{Final Remarks}\label{sec5}

We have constructed several new families of quantum codes with good
as well as asymptotically good parameters. These new quantum codes
have been obtained by applying the CSS construction to classical
algebraic geometry codes constructed here. Many of these codes have
large minimum distances when compared with their code lengths.
Additionally, they have relatively small Singleton defects.
Moreover, we have shown how to obtain sequences of asymptotically
good quantum codes derived from $t$-point AG codes. Therefore, the
class of algebraic geometry codes is a good source to construct
quantum codes with good or even asymptotically good parameters.

\begin{center}
\textbf{Acknowledgements}
\end{center}
I would like to thank the anonymous referees for their valuable
suggestions that improve significantly the quality of this paper. I
also would like to thank the Editor-in-chief Yaakov S. Weinstein for
his excellent work on the review process. This research has been
partially supported by the Brazilian Agencies CAPES and CNPq.

\end{document}